\documentclass[10pt,twoside,twocolumn]{IEEEtran}

\usepackage{cite}
\usepackage[T1]{fontenc}
\usepackage{graphicx}
\usepackage{amssymb}
\usepackage{amsmath}
\usepackage{amsthm}
\usepackage{subfigure}
\usepackage{booktabs} 
\usepackage{multirow}
\usepackage{microtype}
\usepackage{balance}
\usepackage{xcolor}
\usepackage[hidelinks]{hyperref}
\usepackage{mathdots}
\usepackage{epstopdf}

\usepackage{amssymb}
\usepackage{amsfonts}
\usepackage{mathrsfs}
\usepackage{xspace}
\usepackage{bm}
\usepackage{upgreek}

\newcommand{\safemath}[2]{\newcommand{#1}{\ensuremath{#2}\xspace}}

\safemath{\bma}{\mathbf{a}}
\safemath{\bmb}{\mathbf{b}}
\safemath{\bmc}{\mathbf{c}}
\safemath{\bmd}{\mathbf{d}}
\safemath{\bme}{\mathbf{e}}
\safemath{\bmf}{\mathbf{f}}
\safemath{\bmg}{\mathbf{g}}
\safemath{\bmh}{\mathbf{h}}
\safemath{\bmi}{\mathbf{i}}
\safemath{\bmj}{\mathbf{j}}
\safemath{\bmk}{\mathbf{k}}
\safemath{\bml}{\mathbf{l}}
\safemath{\bmm}{\mathbf{m}}
\safemath{\bmn}{\mathbf{n}}
\safemath{\bmo}{\mathbf{o}}
\safemath{\bmp}{\mathbf{p}}
\safemath{\bmq}{\mathbf{q}}
\safemath{\bmr}{\mathbf{r}}
\safemath{\bms}{\mathbf{s}}
\safemath{\bmt}{\mathbf{t}}
\safemath{\bmu}{\mathbf{u}}
\safemath{\bmv}{\mathbf{v}}
\safemath{\bmw}{\mathbf{w}}
\safemath{\bmx}{\mathbf{x}}
\safemath{\bmy}{\mathbf{y}}
\safemath{\bmz}{\mathbf{z}}
\safemath{\bmzero}{\mathbf{0}}
\safemath{\bmone}{\mathbf{1}}

\bmdefine{\biad}{a}
\bmdefine{\bibd}{b}
\bmdefine{\bicd}{c}
\bmdefine{\bidd}{d}
\bmdefine{\bied}{e}
\bmdefine{\bifd}{f}
\bmdefine{\bigd}{g}
\bmdefine{\bihd}{h}
\bmdefine{\biid}{i}
\bmdefine{\bijd}{j}
\bmdefine{\bikd}{k}
\bmdefine{\bild}{l}
\bmdefine{\bimd}{m}
\bmdefine{\bind}{n}
\bmdefine{\biod}{o}
\bmdefine{\bipd}{p}
\bmdefine{\biqd}{q}
\bmdefine{\bird}{r}
\bmdefine{\bisd}{s}
\bmdefine{\bitd}{t}
\bmdefine{\biud}{u}
\bmdefine{\bivd}{v}
\bmdefine{\biwd}{w}
\bmdefine{\bixd}{x}
\bmdefine{\biyd}{y}
\bmdefine{\bizd}{z}

\bmdefine{\bixid}{\xi}
\bmdefine{\bilambdad}{\lambda}
\bmdefine{\bimud}{\mu}
\bmdefine{\bithetad}{\theta}
\bmdefine{\biphid}{\phi}
\bmdefine{\bideltad}{\delta}

\safemath{\bmia}{\biad}
\safemath{\bmib}{\bibd}
\safemath{\bmic}{\bicd}
\safemath{\bmid}{\bidd}
\safemath{\bmie}{\bied}
\safemath{\bmif}{\bifd}
\safemath{\bmig}{\bigd}
\safemath{\bmih}{\bihd}
\safemath{\bmii}{\biid}
\safemath{\bmij}{\bijd}
\safemath{\bmik}{\bikd}
\safemath{\bmil}{\bild}
\safemath{\bmim}{\bimd}
\safemath{\bmin}{\bind}
\safemath{\bmio}{\biod}
\safemath{\bmip}{\bipd}
\safemath{\bmiq}{\biqd}
\safemath{\bmir}{\bird}
\safemath{\bmis}{\bisd}
\safemath{\bmit}{\bitd}
\safemath{\bmiu}{\biud}
\safemath{\bmiv}{\bivd}
\safemath{\bmiw}{\biwd}
\safemath{\bmix}{\bixd}
\safemath{\bmiy}{\biyd}
\safemath{\bmiz}{\bizd}

\safemath{\bmxi}{\bixid}
\safemath{\bmlambda}{\bilambdad}
\safemath{\bmmu}{\bimud}
\safemath{\bmtheta}{\bithetad}
\safemath{\bmphi}{\biphid}
\safemath{\bmdelta}{\bideltad}

\safemath{\bA}{\mathbf{A}}
\safemath{\bB}{\mathbf{B}}
\safemath{\bC}{\mathbf{C}}
\safemath{\bD}{\mathbf{D}}
\safemath{\bE}{\mathbf{E}}
\safemath{\bF}{\mathbf{F}}
\safemath{\bG}{\mathbf{G}}
\safemath{\bH}{\mathbf{H}}
\safemath{\bI}{\mathbf{I}}
\safemath{\bJ}{\mathbf{J}}
\safemath{\bK}{\mathbf{K}}
\safemath{\bL}{\mathbf{L}}
\safemath{\bM}{\mathbf{M}}
\safemath{\bN}{\mathbf{N}}
\safemath{\bO}{\mathbf{O}}
\safemath{\bP}{\mathbf{P}}
\safemath{\bQ}{\mathbf{Q}}
\safemath{\bR}{\mathbf{R}}
\safemath{\bS}{\mathbf{S}}
\safemath{\bT}{\mathbf{T}}
\safemath{\bU}{\mathbf{U}}
\safemath{\bV}{\mathbf{V}}
\safemath{\bW}{\mathbf{W}}
\safemath{\bX}{\mathbf{X}}
\safemath{\bY}{\mathbf{Y}}
\safemath{\bZ}{\mathbf{Z}}

\safemath{\bZero}{\mathbf{0}}
\safemath{\bOne}{\mathbf{1}}
\safemath{\bDelta}{\mathbf{\Delta}}
\safemath{\bLambda}{\mathbf{\UpLambda}}
\safemath{\bPhi}{\mathbf{\Upphi}}
\safemath{\bSigma}{\mathbf{\Upsigma}}
\safemath{\bOmega}{\mathbf{\Upomega}}
\safemath{\bTheta}{\mathbf{\Uptheta}}

\bmdefine{\biAd}{A}
\bmdefine{\biBd}{B}
\bmdefine{\biCd}{C}
\bmdefine{\biDd}{D}
\bmdefine{\biEd}{E}
\bmdefine{\biFd}{F}
\bmdefine{\biGd}{G}
\bmdefine{\biHd}{H}
\bmdefine{\biId}{I}
\bmdefine{\biJd}{J}
\bmdefine{\biKd}{K}
\bmdefine{\biLd}{L}
\bmdefine{\biMd}{M}
\bmdefine{\biOd}{N}
\bmdefine{\biPd}{O}
\bmdefine{\biQd}{P}
\bmdefine{\biRd}{R}
\bmdefine{\biSd}{S}
\bmdefine{\biTd}{T}
\bmdefine{\biUd}{U}
\bmdefine{\biVd}{V}
\bmdefine{\biWd}{W}
\bmdefine{\biXd}{X}
\bmdefine{\biYd}{Y}
\bmdefine{\biZd}{Z}

\bmdefine{\biDelta}{\Delta}
\bmdefine{\biLambda}{\Lambda}
\bmdefine{\biPhi}{\Phi}
\bmdefine{\biSigma}{\Sigma}
\bmdefine{\biOmega}{\Omega}
\bmdefine{\biTheta}{\Theta}

\safemath{\bimA}{\biAd}
\safemath{\bimB}{\biBd}
\safemath{\bimC}{\biCd}
\safemath{\bimD}{\biDd}
\safemath{\bimE}{\biEd}
\safemath{\bimF}{\biFd}
\safemath{\bimG}{\biGd}
\safemath{\bimH}{\biHd}
\safemath{\bimI}{\biId}
\safemath{\bimJ}{\biJd}
\safemath{\bimK}{\biKd}
\safemath{\bimL}{\biLd}
\safemath{\bimM}{\biMd}
\safemath{\bimN}{\biNd}
\safemath{\bimO}{\biOd}
\safemath{\bimP}{\biPd}
\safemath{\bimQ}{\biQd}
\safemath{\bimR}{\biRd}
\safemath{\bimS}{\biSd}
\safemath{\bimT}{\biTd}
\safemath{\bimU}{\biUd}
\safemath{\bimV}{\biVd}
\safemath{\bimW}{\biWd}
\safemath{\bimX}{\biXd}
\safemath{\bimY}{\biYd}
\safemath{\bimZ}{\biZd}

\safemath{\bimDelta}{\biDelta}
\safemath{\bimLambda}{\biLambda}
\safemath{\bimPhi}{\biPhi}
\safemath{\bimSigma}{\biSigma}
\safemath{\bimOmega}{\biOmega}
\safemath{\bimTheta}{\biTheta}

\safemath{\setA}{\mathcal{A}}
\safemath{\setB}{\mathcal{B}}
\safemath{\setC}{\mathcal{C}}
\safemath{\setD}{\mathcal{D}}
\safemath{\setE}{\mathcal{E}}
\safemath{\setF}{\mathcal{F}}
\safemath{\setG}{\mathcal{G}}
\safemath{\setH}{\mathcal{H}}
\safemath{\setI}{\mathcal{I}}
\safemath{\setJ}{\mathcal{J}}
\safemath{\setK}{\mathcal{K}}
\safemath{\setL}{\mathcal{L}}
\safemath{\setM}{\mathcal{M}}
\safemath{\setN}{\mathcal{N}}
\safemath{\setO}{\mathcal{O}}
\safemath{\setP}{\mathcal{P}}
\safemath{\setQ}{\mathcal{Q}}
\safemath{\setR}{\mathcal{R}}
\safemath{\setS}{\mathcal{S}}
\safemath{\setT}{\mathcal{T}}
\safemath{\setU}{\mathcal{U}}
\safemath{\setV}{\mathcal{V}}
\safemath{\setW}{\mathcal{W}}
\safemath{\setX}{\mathcal{X}}
\safemath{\setY}{\mathcal{Y}}
\safemath{\setZ}{\mathcal{Z}}
\safemath{\emptySet}{\varnothing}

\safemath{\colA}{\mathscr{A}}
\safemath{\colB}{\mathscr{B}}
\safemath{\colC}{\mathscr{C}}
\safemath{\colD}{\mathscr{D}}
\safemath{\colE}{\mathscr{E}}
\safemath{\colF}{\mathscr{F}}
\safemath{\colG}{\mathscr{G}}
\safemath{\colH}{\mathscr{H}}
\safemath{\colI}{\mathscr{I}}
\safemath{\colJ}{\mathscr{J}}
\safemath{\colK}{\mathscr{K}}
\safemath{\colL}{\mathscr{L}}
\safemath{\colM}{\mathscr{M}}
\safemath{\colN}{\mathscr{N}}
\safemath{\colO}{\mathscr{O}}
\safemath{\colP}{\mathscr{P}}
\safemath{\colQ}{\mathscr{Q}}
\safemath{\colR}{\mathscr{R}}
\safemath{\colS}{\mathscr{S}}
\safemath{\colT}{\mathscr{T}}
\safemath{\colU}{\mathscr{U}}
\safemath{\colV}{\mathscr{V}}
\safemath{\colW}{\mathscr{W}}
\safemath{\colX}{\mathscr{X}}
\safemath{\colY}{\mathscr{Y}}
\safemath{\colZ}{\mathscr{Z}}

\safemath{\opA}{\mathbb{A}}
\safemath{\opB}{\mathbb{B}}
\safemath{\opC}{\mathbb{C}}
\safemath{\opD}{\mathbb{D}}
\safemath{\opE}{\mathbb{E}}
\safemath{\opF}{\mathbb{F}}
\safemath{\opG}{\mathbb{G}}
\safemath{\opH}{\mathbb{H}}
\safemath{\opI}{\mathbb{I}}
\safemath{\opJ}{\mathbb{J}}
\safemath{\opK}{\mathbb{K}}
\safemath{\opL}{\mathbb{L}}
\safemath{\opM}{\mathbb{M}}
\safemath{\opN}{\mathbb{N}}
\safemath{\opO}{\mathbb{O}}
\safemath{\opP}{\mathbb{P}}
\safemath{\opQ}{\mathbb{Q}}
\safemath{\opR}{\mathbb{R}}
\safemath{\opS}{\mathbb{S}}
\safemath{\opT}{\mathbb{T}}
\safemath{\opU}{\mathbb{U}}
\safemath{\opV}{\mathbb{V}}
\safemath{\opW}{\mathbb{W}}
\safemath{\opX}{\mathbb{X}}
\safemath{\opY}{\mathbb{Y}}
\safemath{\opZ}{\mathbb{Z}}
\safemath{\opZero}{\mathbb{O}}
\safemath{\identityop}{\opI}

\safemath{\veca}{\bma}
\safemath{\vecb}{\bmb}
\safemath{\vecc}{\bmc}
\safemath{\vecd}{\bmd}
\safemath{\vece}{\bme}
\safemath{\vecf}{\bmf}
\safemath{\vecg}{\bmg}
\safemath{\vech}{\bmh}
\safemath{\veci}{\bmi}
\safemath{\vecj}{\bmj}
\safemath{\veck}{\bmk}
\safemath{\vecl}{\bml}
\safemath{\vecm}{\bmm}
\safemath{\vecn}{\bmn}
\safemath{\veco}{\bmo}
\safemath{\vecp}{\bmp}
\safemath{\vecq}{\bmq}
\safemath{\vecr}{\bmr}
\safemath{\vecs}{\bms}
\safemath{\vect}{\bmt}
\safemath{\vecu}{\bmu}
\safemath{\vecv}{\bmv}
\safemath{\vecw}{\bmw}
\safemath{\vecx}{\bmx}
\safemath{\vecy}{\bmy}
\safemath{\vecz}{\bmz}

\safemath{\veczero}{\bmzero}
\safemath{\vecone}{\bmone}
\safemath{\vecxi}{\bmxi}
\safemath{\veclambda}{\bmlambda}
\safemath{\vecmu}{\bmmu}
\safemath{\vectheta}{\bmtheta}
\safemath{\vecphi}{\bmphi}
\safemath{\vecdelta}{\bmdelta}

\safemath{\matA}{\bA}
\safemath{\matB}{\bB}
\safemath{\matC}{\bC}
\safemath{\matD}{\bD}
\safemath{\matE}{\bE}
\safemath{\matF}{\bF}
\safemath{\matG}{\bG}
\safemath{\matH}{\bH}
\safemath{\matI}{\bI}
\safemath{\matJ}{\bJ}
\safemath{\matK}{\bK}
\safemath{\matL}{\bL}
\safemath{\matM}{\bM}
\safemath{\matN}{\bN}
\safemath{\matO}{\bO}
\safemath{\matP}{\bP}
\safemath{\matQ}{\bQ}
\safemath{\matR}{\bR}
\safemath{\matS}{\bS}
\safemath{\matT}{\bT}
\safemath{\matU}{\bU}
\safemath{\matV}{\bV}
\safemath{\matW}{\bW}
\safemath{\matX}{\bX}
\safemath{\matY}{\bY}
\safemath{\matZ}{\bZ}
\safemath{\matzero}{\bmzero}

\safemath{\matDelta}{\bDelta}
\safemath{\matLambda}{\bLambda}
\safemath{\matPhi}{\bPhi}
\safemath{\matSigma}{\bSigma}
\safemath{\matOmega}{\bOmega}
\safemath{\matTheta}{\bTheta}

\safemath{\matidentity}{\matI}
\safemath{\matone}{\matO}

\safemath{\rnda}{A}
\safemath{\rndb}{B}
\safemath{\rndc}{C}
\safemath{\rndd}{D}
\safemath{\rnde}{E}
\safemath{\rndf}{F}
\safemath{\rndg}{G}
\safemath{\rndh}{H}
\safemath{\rndi}{I}
\safemath{\rndj}{J}
\safemath{\rndk}{K}
\safemath{\rndl}{L}
\safemath{\rndm}{M}
\safemath{\rndn}{N}
\safemath{\rndo}{O}
\safemath{\rndp}{P}
\safemath{\rndq}{Q}
\safemath{\rndr}{R}
\safemath{\rnds}{S}
\safemath{\rndt}{T}
\safemath{\rndu}{U}
\safemath{\rndv}{V}
\safemath{\rndw}{W}
\safemath{\rndx}{X}
\safemath{\rndy}{Y}
\safemath{\rndz}{Z}

\safemath{\rveca}{\bimA}
\safemath{\rvecb}{\bimB}
\safemath{\rvecc}{\bimC}
\safemath{\rvecd}{\bimD}
\safemath{\rvece}{\bimE}
\safemath{\rvecf}{\bimF}
\safemath{\rvecg}{\bimG}
\safemath{\rvech}{\bimH}
\safemath{\rveci}{\bimI}
\safemath{\rvecj}{\bimJ}
\safemath{\rveck}{\bimK}
\safemath{\rvecl}{\bimL}
\safemath{\rvecm}{\bimM}
\safemath{\rvecn}{\bimN}
\safemath{\rveco}{\bomO}
\safemath{\rvecp}{\bimP}
\safemath{\rvecq}{\bimQ}
\safemath{\rvecr}{\bimR}
\safemath{\rvecs}{\bimS}
\safemath{\rvect}{\bimT}
\safemath{\rvecu}{\bimU}
\safemath{\rvecv}{\bimV}
\safemath{\rvecw}{\bimW}
\safemath{\rvecx}{\bimX}
\safemath{\rvecy}{\bimY}
\safemath{\rvecz}{\bimZ}

\safemath{\rvecxi}{\bmxi}
\safemath{\rveclambda}{\bmlambda}
\safemath{\rvecmu}{\bmmu}
\safemath{\rvectheta}{\bmtheta}
\safemath{\rvecphi}{\bmphi}

\safemath{\rmatA}{\bimA}
\safemath{\rmatB}{\bimB}
\safemath{\rmatC}{\bimC}
\safemath{\rmatD}{\bimD}
\safemath{\rmatE}{\bimE}
\safemath{\rmatF}{\bimF}
\safemath{\rmatG}{\bimG}
\safemath{\rmatH}{\bimH}
\safemath{\rmatI}{\bimI}
\safemath{\rmatJ}{\bimJ}
\safemath{\rmatK}{\bimK}
\safemath{\rmatL}{\bimL}
\safemath{\rmatM}{\bimM}
\safemath{\rmatN}{\bimN}
\safemath{\rmatO}{\bimO}
\safemath{\rmatP}{\bimP}
\safemath{\rmatQ}{\bimQ}
\safemath{\rmatR}{\bimR}
\safemath{\rmatS}{\bimS}
\safemath{\rmatT}{\bimT}
\safemath{\rmatU}{\bimU}
\safemath{\rmatV}{\bimV}
\safemath{\rmatW}{\bimW}
\safemath{\rmatX}{\bimX}
\safemath{\rmatY}{\bimY}
\safemath{\rmatZ}{\bimZ}

\safemath{\rmatDelta}{\bimDelta}
\safemath{\rmatLambda}{\bimLambda}
\safemath{\rmatPhi}{\bimPhi}
\safemath{\rmatSigma}{\bimSigma}
\safemath{\rmatOmega}{\bimOmega}
\safemath{\rmatTheta}{\bimTheta}

\usepackage{amssymb}
\usepackage{amsfonts}
\usepackage{mathrsfs}
\usepackage{xspace}
\usepackage{bm}
\usepackage{fancyref}
\usepackage{textcomp}

\usepackage{multirow}
\usepackage{stmaryrd}

\newenvironment{textbmatrix}{	\setlength{\arraycolsep}{2.5pt}%
								\big[\begin{matrix}}{\end{matrix}\big]%
								\raisebox{0.08ex}{\vphantom{M}}}

\def\be{\begin{equation}}
\def\ee{\end{equation}}
\def\een{\nonumber \end{equation}}
\def\mat{\begin{bmatrix}}
\def\emat{\end{bmatrix}}
\def\btm{\begin{textbmatrix}}
\def\etm{\end{textbmatrix}}

\def\ba#1\ea{\begin{align}#1\end{align}}
\def\bas#1\eas{\begin{align*}#1\end{align*}}
\def\bs#1\es{\begin{split}#1\end{split}}
\def\bg#1\eg{\begin{gather}#1\end{gather}}
\def\bml#1\eml{\begin{multline}#1\end{multline}}
\def\bi#1\ei{\begin{itemize}#1\end{itemize}}

\newcommand{\orth}{\perp}					

\newcommand{\tp}[1]{\ensuremath{#1^{T}}} 		
\newcommand{\herm}[1]{\ensuremath{#1^{H}}} 	

\safemath{\dirac}{\delta}					
\safemath{\krond}{\dirac}					

\safemath{\upto}{\uparrow}
\safemath{\downto}{\downarrow}
\safemath{\iu}{j}							
\safemath{\ev}{\lambda}						
\safemath{\hilseqspace}{l^{2}}				
\newcommand{\banachfunspace}[1]{\setL^{#1}}	
\safemath{\hilfunspace}{\banachfunspace{2}}	

\safemath{\SNR}{\textit{SNR}} 				
\safemath{\PAR}{\textit{PAR}} 				
\safemath{\No}{N_0}							
\safemath{\Es}{E_s}							
\safemath{\Eb}{E_b}							
\safemath{\EbNo}{\frac{\Eb}{\No}}
\safemath{\EsNo}{\frac{\Es}{\No}}

\DeclareMathOperator{\CHop}{\ensuremath{\opH}} 
\safemath{\tvir}{\rndh_{\CHop}}				
\safemath{\tvtf}{\rndl_{\CHop}}				
\safemath{\spf}{\rnds_{\CHop}}				
\safemath{\bff}{H_{\CHop}}					

\safemath{\ircf}{r_{h}}						
\safemath{\tftvcf}{r_{s}}					
\safemath{\tfcf}{r_{l}}						
\safemath{\bfcf}{r_{H}}						

\safemath{\tcorr}{c_h}						
\safemath{\scf}{c_{s}}						
\safemath{\tfcorr}{c_{l}}					
\safemath{\fcorr}{c_{H}}						

\safemath{\mi}{I}							
\safemath{\capacity}{C}						

\safemath{\normal}{\mathcal{N}}			
\safemath{\jpg}{\mathcal{CN}}			
\safemath{\mchain}{\leftrightarrow}		

\safemath{\dB}{\,\mathrm{dB}}
\safemath{\dBm}{\,\mathrm{dBm}}
\safemath{\Hz}{\,\mathrm{Hz}}
\safemath{\kHz}{\,\mathrm{kHz}}
\safemath{\MHz}{\,\mathrm{MHz}}
\safemath{\GHz}{\,\mathrm{GHz}}
\safemath{\s}{\,\mathrm{s}}
\safemath{\ms}{\,\mathrm{ms}}
\safemath{\mus}{\,\mathrm{\text{\textmu}s}}
\safemath{\ns}{\,\mathrm{ns}}
\safemath{\ps}{\,\mathrm{ps}}
\safemath{\meter}{\,\mathrm{m}}
\safemath{\mm}{\,\mathrm{mm}}
\safemath{\cm}{\,\mathrm{cm}}
\safemath{\m}{\,\mathrm{m}}
\safemath{\W}{\,\mathrm{W}}
\safemath{\mW}{\, \mathrm{mW}}
\safemath{\J}{\,\mathrm{J}}
\safemath{\K}{\,\mathrm{K}}
\safemath{\bit}{\,\mathrm{bit}}
\safemath{\nat}{\,\mathrm{nat}}

\safemath{\define}{\triangleq}			

\safemath{\equivalent}{\sim}
\safemath{\distas}{\sim}					
\safemath{\sdiff}{\Delta}				

\safemath{\reals}{\mathbb{R}}
\safemath{\positivereals}{\reals_{+}}
\safemath{\integers}{\mathbb{Z}}
\safemath{\posint}{\integers_{+}}
\safemath{\naturals}{\mathbb{N}}
\safemath{\posnaturals}{\naturals_{+}}
\safemath{\complexset}{\mathbb{C}}
\safemath{\rationals}{\mathbb{Q}}

\newcommand*{\fancyrefapplabelprefix}{app}		
\newcommand*{\fancyrefthmlabelprefix}{thm}		
\newcommand*{\fancyreflemlabelprefix}{lem}		
\newcommand*{\fancyrefcorlabelprefix}{cor}		
\newcommand*{\fancyrefdeflabelprefix}{def}		
\newcommand*{\fancyrefproplabelprefix}{prop}		
\newcommand*{\fancyrefexmpllabelprefix}{exmpl}
\newcommand*{\fancyrefalglabelprefix}{alg}		
\newcommand*{\fancyreftbllabelprefix}{tbl}		

\frefformat{vario}{\fancyrefseclabelprefix}{Sec.~#1}
\frefformat{vario}{\fancyrefthmlabelprefix}{Theorem~#1}
\frefformat{vario}{\fancyreftbllabelprefix}{Table~#1}
\frefformat{vario}{\fancyreflemlabelprefix}{Lemma~#1}
\frefformat{vario}{\fancyrefcorlabelprefix}{Corollary~#1}
\frefformat{vario}{\fancyrefdeflabelprefix}{Definition~#1}
\frefformat{vario}{\fancyreffiglabelprefix}{Fig.~#1}
\frefformat{vario}{\fancyrefapplabelprefix}{Appendix~#1}
\frefformat{vario}{\fancyrefeqlabelprefix}{(#1)}
\frefformat{vario}{\fancyrefproplabelprefix}{Proposition~#1}
\frefformat{vario}{\fancyrefexmpllabelprefix}{Example~#1}
\frefformat{vario}{\fancyrefalglabelprefix}{Algorithm~#1}

 \newtheorem{thm}{Theorem}
 \newtheorem{cor}[thm]{Corollary}   
 \newtheorem{prop}{Proposition}

\safemath{\dictab}{[\,\dicta\,\,\dictb\,]}

\safemath{\ysig}{\bmy}
\safemath{\ysighat}{\hat{\ysig}}
\safemath{\ysigdim}{M}
\safemath{\xsig}{\bmx}
\safemath{\xsigdim}{N}
\safemath{\nx}{n_x}
\safemath{\zsig}{\bmz}
\safemath{\zsigdim}{\ysigdim}
\safemath{\rsig}{\bmr}
\safemath{\Adict}{\bA}
\safemath{\Adicttilde}{\widetilde{\Adict}}
\safemath{\Adictdim}{\outputdim\times\xsigdim}
\safemath{\avec}{\bma}
\safemath{\avectilde}{\tilde{\avec}}
\safemath{\Bdict}{\bB}
\safemath{\Bdicttilde}{\widetilde{\Bdict}}
\safemath{\Cdict}{\bC}
\safemath{\cvec}{\bmc}
\safemath{\Ddict}{\bD}
\safemath{\Ddictdim}{\ysigdim\times\xsigdim}
\safemath{\dvec}{\bmd}
\safemath{\Ddicttilde}{\widetilde{\bD}}
\safemath{\Bonb}{\bB}
\safemath{\bvec}{\bmb}
\safemath{\Bonbdim}{\ysigdim\times\ysigdim}
\safemath{\noise}{\bmn}
\safemath{\noisedim}{\ysigim}
\safemath{\err}{\bme}
\safemath{\errdim}{\ysigdim}
\safemath{\errset}{\setE}
\safemath{\nerr}{n_e}
\safemath{\delop}{\bP_\errset}
\safemath{\delopc}{\bP_{{\errset}^c}}

\safemath{\cplxi}{\imath}
\safemath{\cplxj}{\jmath}

\safemath{\dict}{\matD}
\safemath{\inputdim}{N}		
\safemath{\outputdim}{M}		
\safemath{\sparsity}{S}	
\safemath{\inputdimA}{{N_a}}	
\safemath{\inputdimB}{{N_b}}	
\safemath{\elemA}{{n_a}}	
\safemath{\elemB}{{n_b}}	
\safemath{\resA}{\matR_a}	
\safemath{\resB}{\matR_b}	
\safemath{\subD}{\matS} 
\safemath{\subA}{\matS_a} 
\safemath{\subB}{\matS_b} 
\safemath{\dicta}{\matA} 	
\safemath{\dictb}{\matB} 	
\safemath{\hollowS}{H}
\safemath{\hollowA}{H_a}
\safemath{\hollowB}{H_b}
\safemath{\cross}{Z}
\safemath{\coh}{\mu_d}			
\safemath{\coha}{\mu_a}			
\safemath{\cohb}{\mu_b}			
\safemath{\mubs}{\nu}	
\safemath{\cohm}{\mu_m} 
\safemath{\dictset}{\setD}	
\safemath{\dictsetp}{\dictset(\coh,\coha,\cohb)}	
\safemath{\dictsetgen}{\dictset_\text{gen}}
\safemath{\dictsetgenp}{\dictsetgen(\coh)}
\safemath{\dictsetonb}{\dictset_\text{onb}}
\safemath{\dictsetonbp}{\dictsetonb(\coh)}

\safemath{\leftside}{U}
\safemath{\rightsideA}{R_a}
\safemath{\rightsideB}{R_b}

\safemath{\indexS}{\setI_S} 

\safemath{\na}{n_a}			
\safemath{\nb}{n_b}			
\safemath{\coeffa}{p_i}	
\safemath{\coeffb}{q_j}	
\safemath{\seta}{\setP}		
\safemath{\setb}{\setQ}     
\safemath{\setw}{\setW}	
\safemath{\setz}{\setZ}	
\safemath{\cola}{\veca}		
\safemath{\colb}{\vecb}		
\safemath{\cold}{\vecd}		
\safemath{\inputvec}{\vecx} 	
\safemath{\error}{\vece}	
\safemath{\noiseout}{\vecz} 	
\safemath{\inputvecel}{x}
\safemath{\inputveca}{\vecx_a}
\safemath{\inputvecb}{\vecx_b}
\safemath{\outputvec}{\vecy}	
\safemath{\lambdamin}{\lambda_{\mathrm{min}}}

\safemath{\elltwo}{\ell_2}
\safemath{\ellone}{\ell_1}
\safemath{\ellzero}{\ell_0}
\safemath{\ellinf}{\ell_\infty}
\safemath{\ellinftilde}{\ell_{\widetilde\infty}}
\safemath{\licard}{Z(\coh,\coha,\cohb)}
\safemath{\xsol}{\hat{x}}
\safemath{\xbord}{x_b}		
\safemath{\xstat}{x_s}		
\safemath{\xstatLone}{\tilde{x}_s}
\safemath{\order}{\mathcal{O}} 
\safemath{\scales}{\Theta} 
\safemath{\ones}{\mathbf{1}} 
\safemath{\zeroes}{\mathbf{0}} 
\safemath{\thlone}{\kappa(\coh,\cohb)} 
\safemath{\constoneA}{\delta} 
\safemath{\constoneB}{\epsilon} 
\safemath{\nlarge}{L}				   
\safemath{\sumlarge}{S_\nlarge}
\safemath{\maxlarger}{P_\nlarge}	   
\safemath{\Pzero}{\textrm{P0}}	
\safemath{\Pone}{\textrm{P1}}
\safemath{\vecfir}{\vecw}			 
\safemath{\vecsec}{\vecz}
\safemath{\elvecfir}{w}              
\safemath{\elvecsec}{z}				 
\safemath{\nlargefir}{n}
\safemath{\normout}{\gamma}
\safemath{\auxfun}{h}
\safemath{\supp}{\textrm{supp}}

\safemath{\indexa}{\ell}
\safemath{\indexb}{r}
\safemath{\indexc}{i}
\safemath{\indexd}{j}

\safemath{\project}{P}

\usepackage{framed}

\safemath{\Wmso}{\bW_{\textnormal{MSO}}}
\safemath{\Wpos}{\bW_{\textnormal{POS}}}
\safemath{\Waug}{\bW_{\textnormal{JAU}}}
\safemath{\bdelta}{\mathbf{\Delta}}
\renewcommand{\bSigma}{\mathbf{\Sigma}}

\safemath{\sfp}{\textsf{p}}
\safemath{\sfc}{\textsf{c}}

\IEEEoverridecommandlockouts
\allowdisplaybreaks 

\begin{document}

\title{\LARGE Single-Antenna Jammers in MIMO-OFDM\\Can Resemble Multi-Antenna Jammers}

\author{
	\IEEEauthorblockN{Gian Marti and Christoph Studer}
	\thanks{
	This work was supported in part by an ETH Research Grant. 
	The work of CS was supported in part by the U.S. National Science Foundation (NSF) under grants CNS-1717559 and ECCS-1824379.
	}
	\thanks{The authors are with the Department of Information Technology and Electrical Engineering, 
	ETH Zurich, Switzerland (email: gimarti@ethz.ch, studer@ethz.ch)}
}

\maketitle

\begin{abstract}
	In multiple-input multiple-output (MIMO) wireless systems with \emph{frequency-flat} channels, a single-antenna
	jammer causes receive interference that is confined to a one-dimensional subspace. 
	Such a jammer can thus be nulled using linear spatial filtering at the cost of one degree of freedom.
\emph{Frequency-selective} channels are often transformed into multiple frequency-flat subcarriers with orthogonal frequency-division multiplexing~(OFDM).
We show that when a single-antenna jammer violates the OFDM protocol by not sending a cyclic prefix, the interference received on each subcarrier by a multi-antenna receiver is, in general, not confined to a subspace of dimension \emph{one} (as a single-antenna jammer in a frequency-flat scenario would be), but of dimension~$L$, where $L$ is the jammer's number of channel taps. 
In MIMO-OFDM systems, a single-antenna jammer can therefore resemble an $L$-antenna jammer. 
Simulations corroborate our theoretical results. 
These findings imply that mitigating jammers with large delay spread through linear spatial filtering is infeasible.
We discuss some (im)possibilities for the way forward.  
\end{abstract}

\begin{IEEEkeywords}
Cyclic prefix, jammer mitigation, MIMO, OFDM.
\end{IEEEkeywords}


\section{Introduction}
\IEEEPARstart{J}{ammers} are a pervasive threat to wireless communication~\cite{pirayesh2022jamming}. 
Linear spatial filtering in~multiple-input multiple-output (MIMO) systems was shown to be effective in mitigating jammers and has been studied extensively; see, e.g.,\cite{pirayesh2022jamming, leost2012interference, marti2023maed}.
Even though many modern wireless systems utilize orthogonal frequency-division multiplexing (OFDM) \cite{hwang2008ofdm} to deal with frequency-selective wideband channels, the jammer-mitigation literature has focused mostly on frequency-flat channels~\cite{do18a, akhlaghpasand20a, marti2023maed, marti2021snips}.
References~\cite{leost2012interference, sodagari2012efficient, mah2015improved, yan2016jamming, zeng2017enabling, pirayesh2021jammingbird} consider frequency-selective channels with OFDM, but substitute a frequency-flat input-output model for each subcarrier where the (single-antenna) jammer is modeled as one-dimensional interference per subcarrier.
However, to transform a frequency-selective channel into mutually orthogonal frequency-flat subcarriers, 
OFDM requires the transmitters to prepend a cyclic prefix to each OFDM symbol \cite{hwang2008ofdm}. 
Since jammers are malicious, they might not send a cyclic prefix.
This has been noted before~\cite{gollakota2011clearing, shahriar2014phy, javed2017novel},
but the consequences of violating the cyclic prefix in MIMO-OFDM systems have not been analyzed and are unknown.

This letter fills this gap by showing 
that a single-antenna jammer that violates the cyclic prefix can appear
as $L$-dimensional interference per subcarrier (i.e., like an $L$-antenna jammer), 
where $L$ is the jammer's number of nonzero channel taps. 
The consequences are (i)~that the common practice of modeling jammers as one-dimensional interference
per OFDM-subcarrier is wrong and (ii)~that mitigating jammers with large delay spread via linear 
spatial~filtering is impractical. Possibilities for the way forward are discussed.

\subsection{Notation}
We use uppercase boldface for matrices; $\tp{\bA}$ and $\herm{\bA}$ are the transpose and conjugate transpose of $\bA$, respectively.
Column vectors are denoted in two different ways: as lowercase boldface (e.g., $\bma$), if they represent \emph{spatial} vectors (e.g., corresponding to an antenna array), or as underlined lowercase boldface (e.g.,~$\underline{\bma}$), if they represent
vectors across time/frequency. 
Comma denotes horizontal concatenation (e.g., $[\bma, \bmb]$); semicolon denotes vertical concatenation (e.g., $[\bma;\bmb]$). 
A vector $\bma\in\opC^N$ in reverse order is denoted by its reflected letter $\reflectbox{\bma}=\tp{[a_N,\dots,a_1]}$.
The subspace spanned by $\bma$ is $\text{span}(\bma)$; its orthogonal complement is $\text{span}(\bma)^\orth$.
The $M\!\times\! N$ zero matrix is $\mathbf{0}_{M\times N}$ or~$\mathbf{0}$.
$\bF_N$ is the unitary $N$-point discrete Fourier transform (DFT) matrix; $\oast$ denotes cyclic convolution. 
$\setC\setN(\mathbf{0},\bC)$ is the circularly-symmetric complex Gaussian distribution with covariance $\bC$.

\section{Prerequisites}
\subsection{MIMO Jammer Mitigation in Frequency-Flat Systems} \label{sec:intro_mit}

In MIMO systems with frequency-flat channels, the input-output (I/O) relation between transmitter and 
receiver in presence of a single-antenna jammer can be modeled as 
\begin{align}
	\bmy[k] = \bH\bms[k] + \bmj w[k] + \bmn[k]. \label{eq:flat_io}
\end{align}
Here, $\bmy[k]\in\opC^B$ is the receive signal of a $B$-antenna receiver at sample instant~$k$,
$\bH\in\opC^{B\times U}$ is the frequency-flat channel matrix between the receiver and one or more 
transmitters with a total of $U$ antennas and transmit signal \mbox{$\bms[k]\in\opC^U$,}
\mbox{$\bmj\in\opC^B$} is the frequency-flat channel of the jammer with transmit signal $w[k]\in\opC$, 
and $\bmn[k]\sim\setC\setN(\mathbf{0},\No\bI_B)$ models noise. 
We assume a block fading scenario so that~$\bH$ and~$\bmj$ do not depend on $k$. 

In~\eqref{eq:flat_io}, the receive interference is restricted to the one-dimensional 
subspace $\text{span}(\bmj)\subset\opC^B$. The receiver can therefore mitigate the jammer by projecting the receive 
signal onto the orthogonal complement $\text{span}(\bmj)^\orth$ of $\text{span}(\bmj)$.
Specifically, let $\bmu_1,\dots,\bmu_{B-1}\in\opC^B$ be an orthonormal basis of $\text{span}(\bmj)^\orth$, 
and let $\bU \triangleq [\bmu_1,\dots,\bmu_{B-1}]\in\opC^{(B-1)\times B}$. Then, the receiver can project
the receive signal onto $\text{span}(\bmj)^\orth$ by computing 
\begin{align} 
	\bar\bmy[k] &\triangleq \herm{\bU}\bmy[k]
	= \underbrace{\herm{\bU} \bH}_{\triangleq\bar\bH\in\opC^{(B-1)\times U}\hspace{-17mm}}\bms[k] 
	~+~ \underbrace{\herm{\bU}\bmj}_{=\mathbf{0}\hspace{-2mm}}w[k]
	+ \underbrace{\herm{\bU}\bmn[k]}_{\triangleq \bar\bmn[k]\hspace{-5mm}} \\
	&= \bar\bH\bms[k] + \bar\bmn[k] \in\opC^{B-1}.
\end{align}
Thus, the receiver can eliminate the jammer at the cost of one degree of freedom, 
obtaining a virtual jammer-free I/O relationship with channel matrix $\bar\bH\in\opC^{(B-1)\times U}$
and noise $\bar\bmn[k]\sim\setC\setN(\mathbf{0},\No\bI_{B-1})$.
The following result is \mbox{well-known}:
\begin{prop}
In frequency-flat MIMO, the receive interference of single-antenna jammer is confined 
to a one-dimensional subspace and can be nulled at the cost of one degree of freedom. 
\end{prop}

The extension to multiple or multi-antenna jammers is straightforward: 
The receiver can completely null jammers with a channel matrix $\bJ$ at the cost
of $\text{rank}(\bJ)$ degrees of~freedom. 

\subsection{OFDM Basics}
We give a short but rigorous derivation of OFDM that serves as 
a template for the proof of our main result in \fref{sec:main}.
OFDM is widely used for eliminating inter-symbol interference (ISI) that results from 
high symbol-rate communication over frequency-selective wideband channels~\cite{hwang2008ofdm}.
OFDM exploits the fact that the DFT matrix diagonalizes circulant matrices 
(i.e., if $\bC$ is circulant, then $\bF_N\bC\herm{\bF_N}$ is diagonal)
to partition a wideband channel into multiple frequency-flat narrowband subcarriers.  
The subcarriers are orthogonal and can be treated~independently,
without need for complex equalization~techniques.

Specifically, consider a discrete-time single-input single-output (SISO) system in which the channel impulse 
response has length $L$ and is denoted by $\underline{\bmh}=\tp{[h[1],\dots,h[L]]}$. 
If the transmitter transmits an information-carrying sequence $\underline{\bmx}=\tp{[x[1],\dots,x[N]]}$, 
then the channel output is a length-$(N+L-1)$ sequence $\underline{\bmy}=\tp{[y[1],\dots,y[N+L-1]]}$ given by 
\begingroup
\setlength{\arraycolsep}{0pt}
\begin{align}
\begin{bmatrix}
	y[1] \\ \vdots \\ y[N\!+\!L\!-\!1]
\end{bmatrix}
= 
\begin{bmatrix}
	h[1]	& 0		& 		& \cdots	& 0		\\[-1mm]
			& \ddots&	 	&		& \vdots	\\[-1mm]
	0		& \cdots& ~\tp{\underline{\reflectbox{\bmh}}} & \cdots & 0 \\[-1mm]
	\vdots	& 		& 		& \ddots& \vdots	\\[-1mm]
	0		&  \cdots& 		& 0		& h[L]
\end{bmatrix}
\begin{bmatrix}
	x[1] \\ \vdots \\ x[N]
\end{bmatrix} 
+ \underline{\bmn}, \label{eq:no_cp_io}
\end{align}
\endgroup
where the entries of the noise $\underline{\bmn}$ are i.i.d.\ $\setC\setN(0,\No)$ distributed.

In OFDM, the transmitter does not simply transmit the information-carrying sequence $\underline{\bmx}$. 
Instead, it prepends a so-called cyclic prefix of length $P\geq L-1$, which consists of the last $P$ symbols of $\underline{\bmx}$. 
That is, the transmitter does not transmit the length-$N$ sequence $\underline{\bmx}$ but the length-$(N+P)$ sequence 
$\tilde{\underline{\bmx}}=\tp{[x[N\!-\!P\!+\!1], \dots, x[N], x[1], x[2], \dots, x[N]]}$. 
The sequence $\tilde{\underline{\bmx}}$ is called an OFDM symbol.  
Analogous to~\eqref{eq:no_cp_io}, the receive sequence associated to an OFDM symbol 
has length $N+P+L-1$. The receiver discards the first $P$ of these receive symbols (which correspond to the cyclic prefix) 
as well as the last $L-1$ receive symbols (which correspond to the reverberation of the channel impulse response, cf. \eqref{eq:no_cp_io}).\footnote{In practice, 
the cyclic prefix of the next OFDM symbol can already be transmitted during the reverberation period, i.e., directly after $x[N]$.}
What remains is a length-$N$ sequence $\underline{\tilde \bmy}=\tp{[\tilde y[1],\dots,\tilde y[N]]}$ whose dependence on the 
transmit signal can be written~as
\begingroup
\setlength{\arraycolsep}{1pt}
\begin{align}
\!\!\!\!\begin{bmatrix}
	\tilde y[1] \\ \vdots \\ \tilde y[N]
\end{bmatrix}
&= 
\begin{bmatrix}
	0&\cdots~0& & \tp{\reflectbox{\underline{\bmh}}}	& 0		& 		&  \cdots	& 0		\\[-1mm]
	& &		& & \ddots&	 	&  		& \vdots	\\
	& &		& & 0& \tp{\reflectbox{\underline{\bmh}}} & 0 & \\[-1mm]
	\vdots& & & & 		& 		& \ddots&   \vdots	\\
	0&\cdots & & & 	& & 0		& \tp{\reflectbox{\underline{\bmh}}}
\end{bmatrix}
\tilde{\underline{\bmx}} 
+ \underline{\bmn} \label{eq:pre_cp_trick} \\
&\hspace{-10mm}= 
\begin{bmatrix}
	h[1] & 0		&  \cdots	& 0	& h[L] & \cdots & & h[2]	\\
	h[2] & h[1]		&  \cdots	& 0	& 0 	   & h[L] & \cdots & h[3]	\\[-1mm]
		& &	  \ddots	&  		& 	\\[-1mm]
	0		&\cdots &	& & h[L] &   \multicolumn{2}{c}{\cdots} & h[1] 
\end{bmatrix}
\begin{bmatrix}
	x[1] \\ \vdots \\ x[N]
\end{bmatrix}
\!+ \underline{\bmn},\! \label{eq:post_cp_trick}
\end{align}
\endgroup
where the step from \eqref{eq:pre_cp_trick} to \eqref{eq:post_cp_trick} exploits the cyclic prefix of 
$\tilde{\underline{\bmx}}$. 
(In~\eqref{eq:pre_cp_trick}, the $P-L+1$ first columns of the matrix are zero.)
Since the matrix in \eqref{eq:post_cp_trick} is circulant, we can also write 
\eqref{eq:post_cp_trick} as a cyclic convolution with the channel impulse~response: 
\begin{align}
	\underline{\tilde \bmy}=\underline{\bmh}\oast\underline{\bmx}+\underline{\bmn}. \label{eq:cyc_conv}
\end{align} 
By defining 
{$\hat{\underline{\bmy}}=\bF_N\underline{\tilde\bmy}$}, $\hat{\underline{\bmh}}=\bF_N[\underline{\bmh};\mathbf{0}_{(N-L)\times1}]$, 
$\hat{\underline{\bmx}}=\bF_N\underline{\bmx}$, and $\hat{\underline{\bmn}}=\bF_N\underline{\bmn}$,
we can restate \eqref{eq:cyc_conv} in the frequency domain~as 
\begin{align}
	\hat{\underline{\bmy}} = \sqrt{N} \text{diag}(\hat{\underline{\bmh}})\,\hat{\underline{\bmx}} + \hat{\underline{\bmn}}.
\end{align}
So if the transmitter sets the information-carrying sequence $\underline{\bmx}$ to the inverse DFT of 
the symbol sequence $\underline{\bms}=\tp{[s[1],\dots,s[N]]}$, $\underline{\bmx}=\herm{\bF_N}\underline{\bms}$, 
and if the receiver ignores the first $P$ receive symbols and computes the DFT of the 
next $N$ receive symbols, $\hat{\underline{\bmy}}=\bF_N\underline{\tilde\bmy}$, then the resulting I/O relation can be stated as 
\begin{align}
	\hat{\underline{\bmy}} =  \sqrt{N} \text{diag}(\hat{\underline{\bmh}})\,\underline{\bms} + \hat{\underline{\bmn}}. \label{eq:ofdm}
\end{align}
{Note the absence of ISI in \eqref{eq:ofdm}, which consists of $N$ scalar I/O relations that correspond}
to $N$ OFDM subcarriers. These subcarrier I/O relations can also be stated individually~as
\begin{align}
	\hat y[k] =  \sqrt{N} \hat h[k] s[k] + \hat n[k], \quad k=1,\dots,N. \label{eq:siso_sc}
\end{align}
This SISO-OFDM template can be applied analogously for SIMO or MIMO systems, 
where the I/O relation for the $k$th subcarrier is 
\begin{align}
	\boldsymbol{\hat}\bmy[k] =  \sqrt{N}  \boldsymbol{\hat}\bmh[k]\bms[k] + \boldsymbol{\hat}\bmn[k], \quad k=1,\dots,N,
\end{align}
and
\begin{align}
	\boldsymbol{\hat}\bmy[k] =  \sqrt{N} \boldsymbol{\hat}\bH[k]\bms[k] + \boldsymbol{\hat}\bmn[k], \quad k=1,\dots,N,
\end{align}
respectively, and where the individual entries of $\boldsymbol{\hat}\bmh[k]$ and $\boldsymbol{\hat}\bH[k]$ correspond to 
the $k$th subcarrier between the individual transmit and receive antennas. 
In both cases, the I/O relation between the individual transmit and receive antennas has the form of~\eqref{eq:siso_sc}.
However, OFDM relies on the transmitter(s) to send a cyclic prefix that enables treating the 
convolution~\eqref{eq:pre_cp_trick} as a cyclic convolution~\eqref{eq:cyc_conv} and to diagonalize it using a DFT matrix. 

\section{Main Result} \label{sec:main}
\emph{If} a single-antenna jammer were to comply with OFDM by sending a cyclic prefix,\footnote{The jammer's 
cyclic prefix would have to be aligned 
with the cyclic prefix of the legitimate transmitter, meaning that it needs to have the same length~$P$. 
Furthermore, the jammer would need to use the same number of subcarriers~$N$ and be properly time-synchronized to the OFDM scheme.}
then (and only then) the per-subcarrier I/O relation of a MIMO-OFDM system would have the form 
\begin{align}
	\boldsymbol{\hat}\bmy[k] = \sqrt{N}\boldsymbol{\hat}\bH[k]\bms[k] + \sqrt{N}\boldsymbol{\hat}\bmj[k] w[k] + \boldsymbol{\hat}\bmn[k] \label{eq:jammer_mimo_sc}
\end{align}
for $k=1,\dots,N$  (with $\boldsymbol{\hat}\bmj[k]$ being the $k$th subcarrier of the jammer's channel to the receiver), which is structurally equivalent to the frequency-flat I/O relation in \eqref{eq:flat_io}. 
The receive interference would be confined to a one-dimensional subspace and could be nulled at the cost of one degree of freedom. 

By nature, however, jammers are malicious---if they can gain an advantage from violating the cyclic prefix, 
they will~do~so. 
The question is how the absence of a cyclic prefix affects the per-subcarrier I/O relation of a MIMO-OFDM system 
under a jamming attack. The answer is the main result of this letter. 
To formally state the result, we assume that the channel impulse response between the single-antenna jammer and the $B$-antenna
receiver consists of $L$ nonzero 
taps $\bmj[1], \dots, \bmj[L]\in\opC^B$.

\begin{thm} \label{thm:main}
	The per-subcarrier receive interference caused by a single-antenna jammer that attacks a MIMO-OFDM system without
	sending a cyclic prefix spans a subspace whose dimension can be up to $\min\{B,L\}$.
\end{thm}

\begin{proof}

There are two cases: In the case $L>B$, the dimension of the interference subspace is trivially bounded by the number of receive antennas $B$, and thus by $\min\{B,L\}$. 
The other case is the case $L \leq B$, for which we will prove that  
the dimension of the interference subspace is bounded by $L$, and thus by $\min\{B,L\}$.
To analyze the receive interference of a single-antenna jammer, we can ignore noise as well as signals sent 
by the legitimate transmitter.  
We use $\underline{\bmj}_b = \tp{[j_b[1],\dots,j_b[L]]}$ to denote the channel impulse response from 
the jammer to the $b$th receive antenna, $b=1,\dots,B$. 

We assume that the attacked communication system (but not the jammer!) uses an OFDM scheme with $N$ subcarriers and a 
cyclic prefix of length $P\geq L-1$.
In accordance with OFDM, the BS aggregates (per~antenna) the receive symbols that correspond
to an OFDM symbol while discarding the portions that correspond to the cyclic prefix and to the reverberation
of the channel impulse response. We can thus write the receive signal at the $b$th antenna corresponding 
to an OFDM symbol~as
\begingroup
\setlength{\arraycolsep}{1pt}
\begin{align}
\begin{bmatrix}
	y_b[1] \\ \vdots \\ y_b[N]
\end{bmatrix}
= 
\begin{bmatrix}
	0~ & \cdots & 0~ & \tp{\reflectbox{\underline{\bmj}}_b} & 0 & \cdots & 0 \\
	\vdots & & & 0 & \tp{\reflectbox{\underline{\bmj}}_b} & \cdots & 0 \\
	& & &  & & \ddots & \vdots \\
	0 & & &  & \cdots & 0 & \tp{\reflectbox{\underline{\bmj}}_b}
\end{bmatrix}
\begin{bmatrix}
	w[-P+1] \\[-2mm] \vdots \\ w[0] \\ w[1] \\[-1mm] \vdots \\ w[N]
\end{bmatrix}
\end{align}
\endgroup
or, more compactly, as 
\begin{align}
	\underline{\bmy}_b &= \bJ_b \underline{\bmw}. \label{eq:time_io}
\end{align}
The matrix $\bJ_b\in\mathbb{C}^{N\times (N+P)}$ is Toeplitz and its first $P-L+1$ columns are zero. 
We now rewrite $\underline{\bmw} = [\underline{\bmw}^{(\sfp)}; \underline{\bmw}^{(\ast)}]$, 
where $\underline{\bmw}^{(\sfp)} = [w[-P+1]; \dots; w[0]]$, 
and $\underline{\bmw}^{(\ast)} = [w[1]; \dots; w[N]]$.
Let $\underline{\bmw}^{(\sfc)}=[w[N-P+1]; \dots; w[N]; \underline{\bmw}^{(\ast)}]$ be the ``cyclic version'' of $\underline{\bmw}$, 
which is obtained by replacing $\underline{\bmw}^{(\sfp)}$ with the last $P$ entries of $\underline{\bmw}^{(\ast)}$. 
We have
\begin{align}
	\underline{\bmw} &= \underline{\bmw}^{(\sfc)} + \underline{\bmw} - \underline{\bmw}^{(\sfc)}\\
	&= \underline{\bmw}^{(\sfc)} + \underline{\Delta}, 
\end{align}
where $\underline{\Delta} \triangleq \underline{\bmw}-\underline{\bmw}^{(\sfc)}$ is nonzero only in the first $P$ entries which are denoted by 
$\underline{\Delta}^{(\sfp)} \in \opC^P$. 
This allows us to rewrite~\eqref{eq:time_io}~as 
\begin{align}
	\underline{\bmy}_b &= \bJ_b \underline{\bmw}^{(\sfc)} + \bJ_b \underline{\Delta} \\
	&= \bJ_b \underline{\bmw}^{(\sfc)} + \bJ^{(\sfp)}_b \underline{\Delta}^{(\sfp)}, \label{eq:before_circulant}
\end{align}
where the matrix $\bJ_b^{(\sfp)}$ consists of the first $P$ columns of $\bJ_b$.
Since $\underline{\bmw}^{(\sfc)}$ is now cyclic, we can define $\hat{\underline{\bmw}}^{(\ast)}=\bF_N\underline{\bmw}^{(\ast)}$, 
$\hat{\underline{\bmy}}_b=\bF_N\underline{\bmy}_b$, $\hat{\underline{\bmj}}_b=\bF_N\underline{\bmj}_b$ and use
the same trick that underlies OFDM: We multiply with the DFT matrix $\bF_N$ on both sides of \eqref{eq:before_circulant}
and restate it in the frequency domain as
\begin{align}
	\hat{\underline{\bmy}}_b= \sqrt{N}\text{diag}\big(\hat{\underline{\bmj}}_b\big) \hat{\underline{\bmw}}^{(\ast)} + \bF_N\bJ^{(\sfp)}_b \underline{\Delta}^{(\sfp)}. 
\end{align}
Focusing on the $k$th subcarrier, we have
\begin{align}
	\hat{y}_b[k]= \sqrt{N}\hat{j}_b[k] \,\hat{w}^{(\ast)}[k] + \tp{\bmf}[k]\bJ^{(\sfp)}_b \underline{\Delta}^{(\sfp)},
\end{align}
where $\tp{\bmf}[k]$ is the $k$th row of $\bF_N$, whose entries are denoted $\tp{\bmf}[k]=[f_1[k],\dots,f_N[k]]$.
If we aggregate the receive signal on the $k$th subcarrier over all $B$ receive antennas, we obtain
\begin{align}
\begin{bmatrix}
	\hat{y}_1[k] \\[-1mm] \vdots \\[-1mm] \hat{y}_B[k]
\end{bmatrix}
= \sqrt{N}
\begin{bmatrix}
	\hat{j}_1[k] \\[-1mm] \vdots \\[-1mm] \hat{j}_B[k] 
\end{bmatrix}
\hat{w}^{(\ast)}[k] + 
\begin{bmatrix}
	 \tp{\bmf}[k]\bJ^{(\sfp)}_1 \\[-1mm] \vdots \\[-1mm]  \tp{\bmf}[k]\bJ^{(\sfp)}_B 
\end{bmatrix} \underline{\Delta}^{(\sfp)}, 
\end{align}
or, in vector notation, and by defining the matrix $\bR[k] \triangleq [ \tp{\bmf}[k]\bJ^{(\sfp)}_1; \dots;  \tp{\bmf}[k]\bJ^{(\sfp)}_B]$:
\begin{align}
	\hat{\bmy}[k] &= \sqrt{N}\, \hat{\bmj}[k]\,\hat{w}^{(\ast)}[k] + \bR[k] \underline{\Delta}^{(\sfp)}. \label{eq:sc_jammer}
\end{align}
The second term in \eqref{eq:sc_jammer} is the effect of the jammer violating the cyclic prefix.
Its presence implies that the interference of a single-antenna jammer on a given subcarrier is not constrained to 
a one-dimensional subspace anymore.
We rewrite \eqref{eq:sc_jammer} as 
\begin{align}
\hat{\bmy}[k] &= \big[\sqrt{N}\, \hat{\bmj}[k],~ \bR[k] \big]
	\begin{bmatrix}
		\hat{w}^{(\ast)}[k] \\ \underline{\Delta}^{(\sfp)}
	\end{bmatrix} 
	= \boldsymbol{\mathsf{J}}[k]\boldsymbol{\mathsf{w}}[k], \label{eq:effective_jammer_channel}
\end{align}
where $\boldsymbol{\mathsf{J}}[k]\triangleq\big[\sqrt{N}\, \hat{\bmj}[k],~ \bR[k] \big]\in\opC^{B\times(P+1)}$
serves as the jammer's effective channel matrix on the $k$th subcarrier
and $\boldsymbol{\mathsf{w}}[k]\triangleq[\hat{w}^{(\ast)}[k]; \underline{\Delta}^{(\sfp)}]\in\opC^{P+1}$
serves as the jammer's~effective transmit signal on the $k$th subcarrier.
The rank of $\boldsymbol{\mathsf{J}}[k]$~is
\begin{align}
	\text{rank}(\boldsymbol{\mathsf{J}}[k])\leq L, \label{eq:rank_bound}
\end{align}
where the bound is tight in general (i.e., equality can be achieved). 
To see this, note that $\bR[k]$ can be written as
\begin{align}
	\bR[k] =
	\begin{bmatrix}
		\tp{\bmf}[k] & \mathbf{0} & \!\cdots\! & \mathbf{0} \\[-1mm]
		\mathbf{0} & \tp{\bmf}[k] & & \vdots\\[-1mm]
		\vdots & & \ddots & \mathbf{0} \\
		\mathbf{0} & \cdots & \mathbf{0} & \tp{\bmf}[k] 
	\end{bmatrix}
	\begin{bmatrix}
		\bJ^{(\sfp)}_1 \vspace{1mm} \\  \vdots \vspace{2mm} \\ \bJ^{(\sfp)}_B
	\end{bmatrix}, 
\end{align}
where the left factor has dimensions $B\times BN$ and rank~$B$, 
and the right factor has rank $L-1$ since it is a $BN\times P$ block matrix whose row blocks 
are upper triangular and have form 
\begin{align}
	\bJ_b^{(p)} = 
	\begin{bmatrix}
		0 & \cdots & 0 & j_b[L] & \!\cdots\! & j_b[3] & j_b[2] \\[-1mm]
		\vdots &  &  & 0 & j_b[L] & \!\cdots\! & j_b[3] \\[-1mm]
		\vdots &  &  &  & 0 & \!\ddots\! &  \\[-1mm]
		0 & \cdots &  &  &  & 0 & j_b[L] \\[-1mm]
		0 & \cdots &  &  &  &  \cdots & 0 \\[-1mm]
		\vdots &   &  &  &  &   & \vdots \\[-1mm]
		0 & \cdots &  &  &  &  \cdots & 0
	\end{bmatrix},
\end{align}
where only the last $L-1$ columns are nonzero,
implying the inequality in \eqref{eq:rank_bound}. 
For equality, consider the example where
$\underline{\bmj}_b=\underline{\bme}_b$ (the $b$th standard unit vector, whose entries are all zero 
except for the $b$th entry, which is one) for $b=1,\dots,L$, 
and $\underline{\bmj}_b=\underline{\boldsymbol{0}}$ for $b=L+1,\dots,B$.
In that case, we have \pagebreak[0]
\begin{align}
	\bR_k = 
	\begin{bmatrix}
		0\,\cdots\,0 & 0 & \cdots & 0 & f_1[k] \\[-1mm]
		0\,\cdots\,0 & \vdots & 0 & f_1[k] & f_2[k] \\[-1mm]
		0\,\cdots\,0 & & \iddots \\[-1mm]
		0\,\cdots\,0 & f_1[k] & f_2[k] & \cdots & f_{L-1}[k] \\
		0\,\cdots\,0 & 0 & 0 & \cdots & 0 \\[-1mm]
		\vdots & & & & \vdots \\[-1mm]
		0\,\cdots\,0 & 0 & 0 & \cdots & 0
	\end{bmatrix}, 
\end{align}
which has rank $L-1$. 
So in general, $\bR_k$ can have rank up to $L-1$.
Furthermore, we have $\hat\bmj_k = \tp{[f_1[k],\dots,f_L[k]]}$, which is not 
contained in the columnspace of $\bR_k$, so $\text{rank}(\boldsymbol{\mathsf{J}}[k])=L$.

The fact that $\text{rank}(\boldsymbol{\mathsf{J}}[k])=L$ does not by itself already prove that 
the receive interference can span a subspace of dimension $L$. We also need that the signal 
$\boldsymbol{\mathsf{w}}[k]\in\opC^{P+1}$, which is multiplied by $\boldsymbol{\mathsf{J}}[k]$, can take on any value
in $\opC^{P+1}$. But this follows from the fact that $\hat{w}^{(\ast)}[k]$ is simply the $k$th entry of the 
Fourier transform of the prefix-free jammer signal $\underline{\bmw}^{(\ast)}$ and, thus, can be set to any value by the jammer; 
$\underline{\Delta}^{(\sfp)}$ is simply the difference between the jammer's actual transmit sequence 
$\underline{\bmw}$ and its cyclic version $\underline{\bmw}^{(\sfc)}$ and so can also be set to any value by the jammer. 
From this, the result follows.
\end{proof}

Eq.~\eqref{eq:effective_jammer_channel} shows that a cyclic prefix violating jammer does~not  
yield a per-subcarrier I/O relation as in \eqref{eq:jammer_mimo_sc}, but one of the~form
\begin{align}
	\boldsymbol{\hat}\bmy[k] = \sqrt{N}\boldsymbol{\hat}\bH[k]\bms[k] 
	+ \boldsymbol{\mathsf{J}}[k]\boldsymbol{\mathsf{w}}[k] + \boldsymbol{\hat}\bmn[k], \label{eq:jammer_mimo_eff_sc}	
\end{align}
where $\boldsymbol{\mathsf{J}}[k]$ is a matrix of rank $\min\{B,L\}$. 
By violating the cyclic prefix of OFDM, a single-antenna jammer therefore looks on each subcarrier 
like an $L$-antenna transmitter in a frequency-flat communication system. 
This diminishes the effectiveness of linear multi-antenna processing for jammer-mitigation, 
since nulling the jammer with a linear projection now comes at the cost of $L$  degrees of freedom instead of one.
{The extension to multi-antenna jammers is straightforward (proof~omitted):}
\begin{cor} \label{cor:main}
	If an $I$-antenna jammer (whose $i$th antenna has $L_i$ nonzero channel taps) attacks a MIMO-OFDM system without
	sending a cyclic prefix, then its receive interference spans a subspace whose dimension~can be up to 
	$\min\{B,\sum_{i=1}^I \!L_i\}$.
\end{cor}

\section{Simulation Results}
We now use simulations to demonstrate the practical implications of the increased dimension of the interference subspace.\footnote{
Simulation code for reproducing our results and
for simulating different system parameters is available at \url{https://github.com/IIP-Group/OFDM-jammer}.
}
Specifically, we show that nulling a single dimension (per subcarrier) of the receive signal is sufficient 
for mitigating an OFDM-compliant single-antenna jammer (i.e., a jammer that sends a cyclic prefix), 
but that it is insufficient for mitigating a single-antenna jammer that does \emph{not} send a cyclic prefix 
(since, by \fref{thm:main}, the receive interference can have $\text{rank}(L)>1$). 

\subsection{Simulation Setup}
We consider a MIMO-OFDM communication system with $B=8$ antennas at the receiver and $U=2$ antennas at the transmitter
who is transmitting two independent QPSK data streams.
We consider OFDM as in the $20$\,MHz mode of IEEE 802.11n, with $N=64$ subcarriers 
(of which only $|\setK|=48$ are used for data transmission---we ignore the other subcarriers in our simulations) 
and a cyclic prefix of length $P=16$. We assume a block fading model in which the channels stay constant for 
$M=50$ OFDM symbols.
We assume a Rayleigh fading channel model with $L=4$ taps for the transmitter and the jammer, 
where the entries of the time-domain channel matrices/vectors $\bH[t],\bmj[t],t=1,\dots,L$ 
are i.i.d. $\setC\setN(0,1)$.
The jammer emits i.i.d. circularly-symmetric complex Gaussian noise at $25$\,dB higher receive energy than
the legitimate signal.

\subsection{Mitigation Performance of a Projection with Rank $B-1$}
\fref{fig:comp_noncomp} shows the performance as measured in uncoded bit eror rate (BER) vs. 
signal-to-noise ratio (SNR) for three scenarios. 

First, to serve as a benchmark, a communication system that does \emph{not} suffer from a jammer (jammerless benchmark), 
and in which the receiver uses a zero-forcing (ZF) data detector. 

Second, a communication system attacked by a single-antenna jammer that transmits i.i.d. Gaussian symbols 
but that  \emph{does adhere} to the OFDM protocol by transmitting a cyclic prefix (OFDM-compliant jammer). 
In this system, before the ZF data detector, the receiver uses a rank $B-1$ projection on each subcarrier to 
remove the strongest dimension of the jammer interference (cf. \fref{sec:intro_mit}), which in this case---since the jammer sends a cyclic 
prefix---is the \emph{only} jammer dimension. So the projection nulls the jammer perfectly, and the 
performance loss of $1$\,dB (at $0.1\%$ BER) compared to the jammerless case comes solely from the lost degree of~freedom. 

Third, a communication system attacked by a single-antenna jammer that transmits i.i.d. Gaussian symbols 
\emph{without transmitting a cyclic prefix} (cyclic prefix-violating jammer). As in the previous scenario, the receiver uses a ZF data detector 
that is preceded on each subcarrier by a rank $B-1$ projection that removes the strongest dimension of the jammer interference. 
According to \fref{thm:main}, since the jammer sends no cyclic prefix, the receive interference
is not one-dimensional but in general $(L=4)$-dimensional. Nulling a single dimension should therefore not be 
sufficient for removing the jammer interference. Our simulations confirm this theoretical result, 
since the BER remains above $20\%$ for all SNRs due to the only partially removed jammer interference.

\begin{figure*}[tp]
\centering
\begin{minipage} {0.31\textwidth}
	\includegraphics[height=4cm]{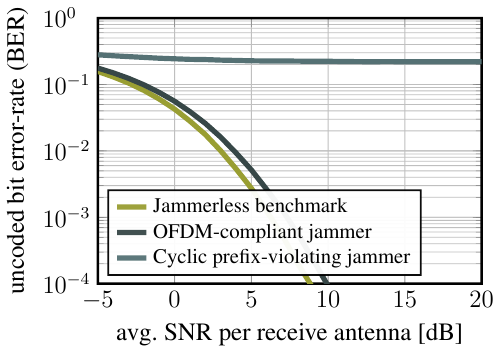}
	\caption{Bit error rates (BERs) of a jammerless system, a system attacked by an OFDM-compliant jammer, 
	and a system attacked by a cyclic prefix-violating jammer. The latter two use a rank $B-1$ orthogonal projection 
	for mitigating the jammer.}
	\label{fig:comp_noncomp}
\end{minipage}
\hspace{2mm}
\begin{minipage} {0.31\textwidth}
	\includegraphics[height=4cm]{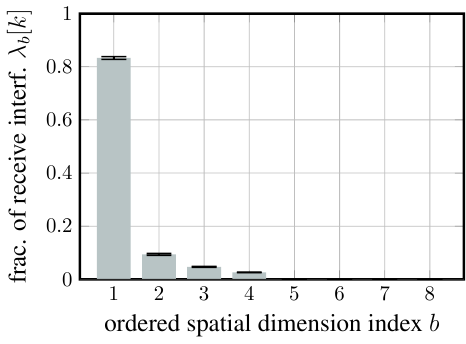}
	\caption{Histogram of the fraction of receive interference  over the different spatial dimensions
	(sorted in descending order of power). The receive interference occupies $L=4$ different
	dimensions. Error bars cover two standard deviations per side.}
	\label{fig:distribution}
\end{minipage}
\hspace{2mm}
\begin{minipage} {0.31\textwidth}
	\includegraphics[height=4cm]{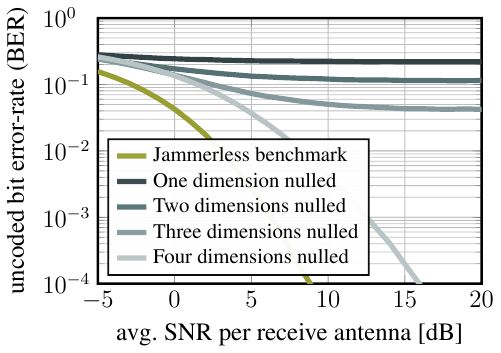}
	\caption{Bit error rate (BER) of a system that is attacked by a cyclic prefix-violating jammer.
	The system uses projections that null different numbers of dimensions (from one to four) of the receive
	interference for jammer mitigation.}
	\label{fig:projectors}
\end{minipage}
\vspace{-2mm}
\end{figure*}

\subsection{Spatial Distribution of the Receive Interference}
That the receive interference of a cyclic prefix-violating jammer occupies an $(L=4)$-dimensional subspace
is {further confirmed} by \fref{fig:distribution}, which shows a normalized histogram of the ordered singular
values of the receive interference (measured across a coherence interval). That is, we denote the receive interference on the
$k$th subcarrier over a coherence interval by
$\hat{{\bY}}_R[k]\in\opC^{B\times M}$ (which consists of $M$ receive vectors) 
and decompose it
using a singular-value decomposition as $\hat{{\bY}}_R[k]=\bU[k]\bSigma[k]\herm{\bV}[k]$, 
where the diagonal elements $\sigma_b[k]$ of $\bSigma[k]$ are sorted in descending order. Then we define
the fraction of receive interference on the $b$th ordered spatial dimension as 
$\lambda_b[k]=\sigma_b[k]/\text{tr}(\bSigma[k])$ for $b=1,\dots,B$. \fref{fig:distribution} shows the mean (plus/minus
two standard deviations) of the distribution 
over these $\lambda_b[k]$ (over all subcarriers $k$ and Monte-Carlo realizations). 
We see that while a single dimension contains a large part of the receive interference, it does not 
contain all of it. As predicted by \fref{thm:main}, the receive interference occupies exactly $L=4$ spatial dimensions.

\subsection{Nulling Different Numbers of Interference Dimensions}
\fref{fig:projectors} shows the performance when a communication system tries to mitigate a
cyclic prefix-violating jammer by nulling different numbers of dimensions (from one to \mbox{$L=4$}) 
of the receive interference using an orthogonal projection before detecting the data using a ZF detector.
In each case, the nulled dimensions are those where the receive interference is strongest.
\fref{fig:projectors} shows that the performance increases with each nulled dimension, as the residual 
jammer interference decreases until, with four nulled dimensions, the interference
is removed completely. However, comparing the performance of this case against the performance of 
nulling a single dimension in the case of an OFDM-compliant jammer (\fref{fig:comp_noncomp}) shows a performance
loss of more than $5$\,dB at a BER of $0.1\%$ (cf. also the large loss compared to the jammerless benchmark). 
This loss comes from the fact that removing the cyclic prefix violating 
jammer comes at the cost of $L=4$ degrees of freedom whereas nulling an OFDM-compliant jammer only costs one degree
of freedom (out of a total of $B=8$ degrees of freedom). 
These results show that, if the jammer's delay spread $L$ is large, jammer mitigation
through linear spatial filtering becomes less effective---if $L$ approaches~$B$, it becomes completely infeasible.

\section{The Way Forward: Some (Im)Possibilities}
This letter has {revealed a difficulty} of MIMO jammer~mitigation in OFDM. {We remark that 
this difficulty is not simply an artifact of OFDM, but is in fact inherent to jammers with 
frequency-selective channels.} Mitigating such a jammer using orthogonal projections in the time-domain would
\emph{also} require nulling of $L$ dimensions (for a single-antenna jammer)---every dimension corresponding to a
tap of the jammer's channel impulse response---at the cost of $L$ degrees of freedom.
{For this reason, one possible way forward for communication 
systems that~pursue jammer resilience based on MIMO processing seems to avoid frequency-selective channels altogether.}
This means that communication has to be
restricted to narrowband channels, either at the cost of achievable data rates, or by using multiple narrowband 
carriers in parallel \cite[Sec. 12.1]{goldsmith05a}.

What this letter has conclusively established is that the common approach of modeling single-antenna jammers with 
frequency-selective channels in OFDM as frequency-flat single-antenna jammers is inaccurate, 
and that the actual high-rank interference reduces the effectiveness of linear spatial filtering.

\vfill 
\linespread{1.01}


\end{document}